\documentclass[preprint]{elsarticle}
\usepackage[utf8]{inputenc}
\usepackage{graphicx}

\newtheorem{theorem}{Theorem}
\newtheorem{lemma}{Lemma}
\newtheorem{corollary}{Corollary}
\newproof{proof}{Proof}

\begin{document}
\begin{frontmatter}
\title{Modeling Dengue Outbreaks}
\author{Marcelo J Otero, Daniel H Barmak, Claudio O Dorso, Hernán G Solari}
\address{Departamento de Física, FCEyN-UBA and IFIBA-CONICET}
\author{Mario A Natiello}
\address{Centre for Mathematical Sciences, Lund University, Sweden}
\begin{abstract}
We introduce a dengue model (SEIR) where the human individuals are
treated on an
individual basis (IBM) while the mosquito population, produced by an
independent model, is treated by compartments (SEI). We study the spread of
epidemics by the sole action of the mosquito. Exponential, deterministic and
experimental distributions for the (human) exposed period are considered in two
weather scenarios, one corresponding to temperate climate and the other to
tropical climate. Virus circulation, final epidemic size and duration of
outbreaks are considered showing that the results present little sensitivity to
the statistics followed by the exposed period provided the median of the
distributions are in coincidence. Only the time between an introduced (imported)
case and the appearance of the first symptomatic secondary case is sensitive to
this distribution. We finally show that the IBM model introduced is precisely
a realization of a compartmental model, and that at least in this case, the
choice between compartmental models or IBM is only a matter of convenience.
\end{abstract} 

\end{frontmatter}

\section{Introduction}
Dengue fever is a vector-born disease produced by a \emph{flavivirus} of the
family \emph{flaviviridae} \cite{gubl98}. The main vectors of dengue are
\emph{Aedes aegypti} and \emph{Aedes albopictus}.

The research aimed at producing dengue models for public policy use began with
Newton and Reiter \cite{newt92} who introduced the minimal model for dengue in
the form of a set of Ordinary Differential Equations (ODE) for the human
population disaggregated in Susceptible, Exposed, Infected and Recovered
compartments. The mosquito population was not modeled in this early work.  A
different starting point was taken by Focks et al. \cite{fock93a,fock93b} that
began by describing mosquito populations in a computer framework named
Dynamic Table Model where later the human population (as well as the disease)
was introduced \cite{fock95}.

Newton and Reiter's model ({\bf NR}) favours economy of resources and
mathematical accessibility, in contrast, Focks' model emphasize realism. These
models represent in Dengue two contrasting compromises in the standard 
trade-off in modeling. A third starting point has been recently added. Otero, Solari
and Schweigmann ({\bf OSS}) developed a dengue model \cite{oter10} which
includes the evolution of the mosquito population \cite{oter06, oter08} and is
spatially explicit. This last model is somewhat in between Focks' and NR as it is
formulated as a state-dependent Poisson model with exponentially distributed
times.

Each approach has been further developed
\cite{este98,este99,este00,bart02,pong03,mago09,fern10}. ODE models have
received most of the attention. Some of the works explore: variability of
vector population \cite{este98}, human population \cite{este99}, the effects of
hypothetical vertical transmission of Dengue in vectors \cite{este00},
seasonality \cite{bart02}, age structure \cite{pong03} as well as incomplete
gamma distributions for the incubation and infectious times \cite{chow07}.
Contrasting modeling outcomes with those of real epidemics has shown the need 
to consider spatial heterogeneity as well \cite{favi05}.

The development of computing technology has made possible to produce Individual
Based Models ({\bf IBM}) for epidemics \cite{grim99,bian04}. IBM have been
advocated as the most realistic models \cite{bian04} since their great
flexibility allows the modeler to describe disease evolution and human mobility
at the individual level. When the results are only to be analysed numerically,
IBM are probably the best choice.  However, they are frequently presented in a
most unfriendly way for mathematicians as they usually lack a formulation
(expression in closed formulae) and are --at best-- presented as algorithms if
not just in words \cite{grim99}. In contrast, working on the ODE side, it has
been possible, for example, to achieve an understanding of the influence of
distribution of the infectious period in epidemic modeling
\cite{lloy01,lloy01b,feng07}. IBM have been used to study the time interval
between primary and secondary cases \cite{fine03} which is influenced, in the
case of dengue, mainly by the extrinsic (mosquito) and intrinsic (human)
incubation period.

In this work an IBM model for human population in a dengue epidemic is
presented. The model is driven by mosquito populations modeled with spatial
heterogeneity with the method introduced in \cite{oter08} (see Section II).
The IBM model is then used to examine the actual influence of the distribution
of the incubation period comparing the most relevant information produced by
dengue models: dependence of the probability of dengue circulation with respect
to the mosquito population and the total epidemic size. Exponential, delta
(fixed times, deterministic) and experimental \cite{nish07} distributions are
contrasted (Section III). The infectious period and the extrinsic incubation
period is modeled using experimental data and measured transmission rates
(human to mosquito) \cite{nish07}.

The IBM model produced is critically discussed. We show that it can be mapped
exactly into a stochastic compartmental model of a novel form (see Section IV)
thus crossing for the first time the valley separating IBM from compartmental
models.  This result opens new perspectives which we also discuss in Section
IV. Finally, Section V presents the conclusions of this work.

\section{The model}
 It is currently accepted that the dengue virus does not make any effect to the
vector.  As such, \emph{Aedes aepgypti} populations are independent of the
presence of the virus.  In the present model mosquito populations are produced
by the \emph{Aedes aegypti} model \cite{oter08} with spatial resolution of one
{\em block} using climatic data tuned to Buenos Aires, a temperate city where
dengue circulated in the summer season 2008-2009 \cite{seij09}. The urban unit
of the city is the {\em block} (approximately a square of 100m x 100m).
Because of the temperate climate the houses are not open as it is often the
case in tropical areas.  Mosquitoes usually develop in the center of the block
which often presents vegetation and communicates the buildings within the
block. The model then assumes that mosquitoes belong to the block and not to
the houses and they blood-feed with equal probability in any human resident in
the block. \emph{Aedes aegypti} is assumed to disperse seeking for places to
lay eggs. The mosquito population, number of bites per day, dispersal flights
and adult mortality information per block is obtained from the mosquito model
\cite{oter08}.

The time-step of  the model has been fixed at one day.  The human population of
each block is fixed in the present work and the disease is spatially spread by
the mosquito alone. The evolution of the disease in one individual human,
$h$, proceeds as follows:
\begin{list}{}{}
\item{Day $d=d_0$} The virus is transmitted by the bite of an infected
mosquito
\item{Day $d=d_0+\tau_E(h)$ } The human $h$ becomes infective ($h$ is said to
be exposed to the virus during this period of time).
\item{Day $d=d_0+\tau_E(h)+j$} For $1 \le j \le \tau_I$ the human $h$ is
infective and transmits the virus to a biting mosquito with a probability
$p_{hm}(j)$. $\tau_I$ indicates the duration in days of the viremic window.
\item{Day $d > d_0+\tau_E(h) + \tau_I $} The human $h$ is recovered and no
longer transmits dengue.
\end{list}
The cycle in the human being is then of the form Susceptible, Exposed,
Infected, Recovered (SEIR).

The virus enters the mosquito when it bites a viremic human with a probability
$p_{hm}(j)$ depending of the day $j$ in the infectious cycle of the specific
human bitten.  The cycle continues with the reproduction of the virus within
the mosquito (extrinsic period), lasting $\tau_m$ days (in this work $\tau_m$ 
was set to 8 days). After this reproduction
period the mosquito becomes infectious and transmits the virus with a
probability $p_mh$ when it bites. The mosquito follows a cycle Susceptible,
Exposed, Infected (SEI) and does not recover.
\cite{gubl98,fock95,nish07,oter10}.  The adult female mosquito population as
produced by the \emph{Aedes aegypti} simulation is then split into susceptible,
$\tau_m$ stages of exposed and one infective compartment according to their
interaction with the viremic human population and the number of days elapsed
since acquiring the virus.

The epidemic starts when one or more humans become viremic. The algorithm
followed is:

\begin{enumerate}
\item Give individual attributes, $\tau_{E}(h)$ according to prescribed
distribution.
\item Read geometry, size of viremic window, probabilities
$p_{mh}(j), \:j\!=\! 1...\tau_I $, human population 
in each block, day of the year when the epidemic starts.

\item Initialise the blocks with the human population.

\item Read total adult female mosquito population of the day ($M$), bites,
flights to neighbouring blocks and mosquito death probability. Initialise
all mosquitoes as susceptible ($MS$). Set infective bites to zero.

\item Day-loop begins:

\begin{enumerate}
\item\label{loop}  Calculate the amount of surviving infected mosquitoes. Compute surviving
mosquito population with $d$ exposed days, age them by one day.
Exposed mosquitoes evolve to infected ones after $d=\tau_m$ days
(Use binomial random number generator).
\item Compute probability for a mosquito bite to transmit dengue as $p_{minf}=p_{mh}M_{I}/M$
(compound probability of being infective and being effective in the
transmission). Each bite is an independent event according to the
underlying mosquito model.
\item Compute the probability for a bite to be made by a susceptible mosquito
out of all the non-infective bites $p_{otras}=M_{S}/(M (1-p_{minf}))$.  The
non-infective bites have probability $(1-p_{minf})$, and come from infected
mosquitoes failing to transmit the virus, exposed and susceptible mosquitoes,
$M (1-p_{minf})=M-p_{mh}M_{I}=M_{S}+M_{E}+(1-p_{mh})M_{I}$.
\item Calculate the number of infected and susceptible mosquito bites
using binomials with the previous probabilities and the number of
total bites.
\item Compute number of humans bitten in each day of 
the infected state,  $H_I(j)$.
\item Compute the number of new exposed mosquitoes, taking into account that the probability
of human-mosquito contagion is dependent on the stage-day of the human
infection. The amount of new infected insects is chosen using binomials. For
this purpose, we add the results of the calculation of the number of
susceptible mosquitoes that bite humans and get infected, 
$M_{NE}=\sum_{j=1}^{\tau_I} Bin(p_{hm}(j),H_I(j))$. 
Where $M_{NE}$ are the new exposed mosquitoes,
$Bin(H_I(j),p_{hm}(j))$ is a binomial realization with the day-dependent
probability $p_{hm}(j)$ and $H_I(j)$ the quantity of infected humans bitten by
susceptible mosquitoes on infection day $j$.
\item\label{tablita} Perform a random equi-distributed selection of humans 
bitten by infectious
mosquitoes. Build a table of bitten individuals.
\item Update the state of all the humans. If the human belongs to the
Susceptible state and has been bitten according to the table built in (\ref{tablita}) 
then change the state of selected human to Exposed, record 
$d_0$ for each exposed human.
Susceptible individuals not bitten by an infective mosquito will remain as such
and consequently their intrinsic time $d$ remains in $0$. Increase the
intrinsic time of Exposed and Infected humans by one. Those Exposed individuals
for which their intrinsic time has surpassed the value $d=d_0+\tau_{E}(h)$ are
moved to the Infected state, while those Infected 
individuals whose intrinsic time
is larger than $d=d_0+\tau_{E}(h)+\tau_{I}$ , are moved to the Removed state
\begin{enumerate}
\item Compute the number of individuals in each state for every cell.
\item Read the total adult female mosquito population of the day ($M$), bites,
flights to neighboring blocks. 
\end{enumerate}
\end{enumerate}
\item Repeat all over again from (\ref{loop}). Each iteration is a new day of
the simulation. 
\end{enumerate}

\section{Epidemic dependence on the distribution of the exposed period}

We implemented four different distributions for the duration of the exposed
period assigned to human individuals: Nishiura's experimental distribution
(\cite{nish07}), a delta and exponential distributions with the same mean that
the experimental one and an exponential distribution with the same median than
the experimental one. We call them N,D,E1 and E2 respectively.
\begin{figure}[f]
\includegraphics[width=8cm,angle=-90]{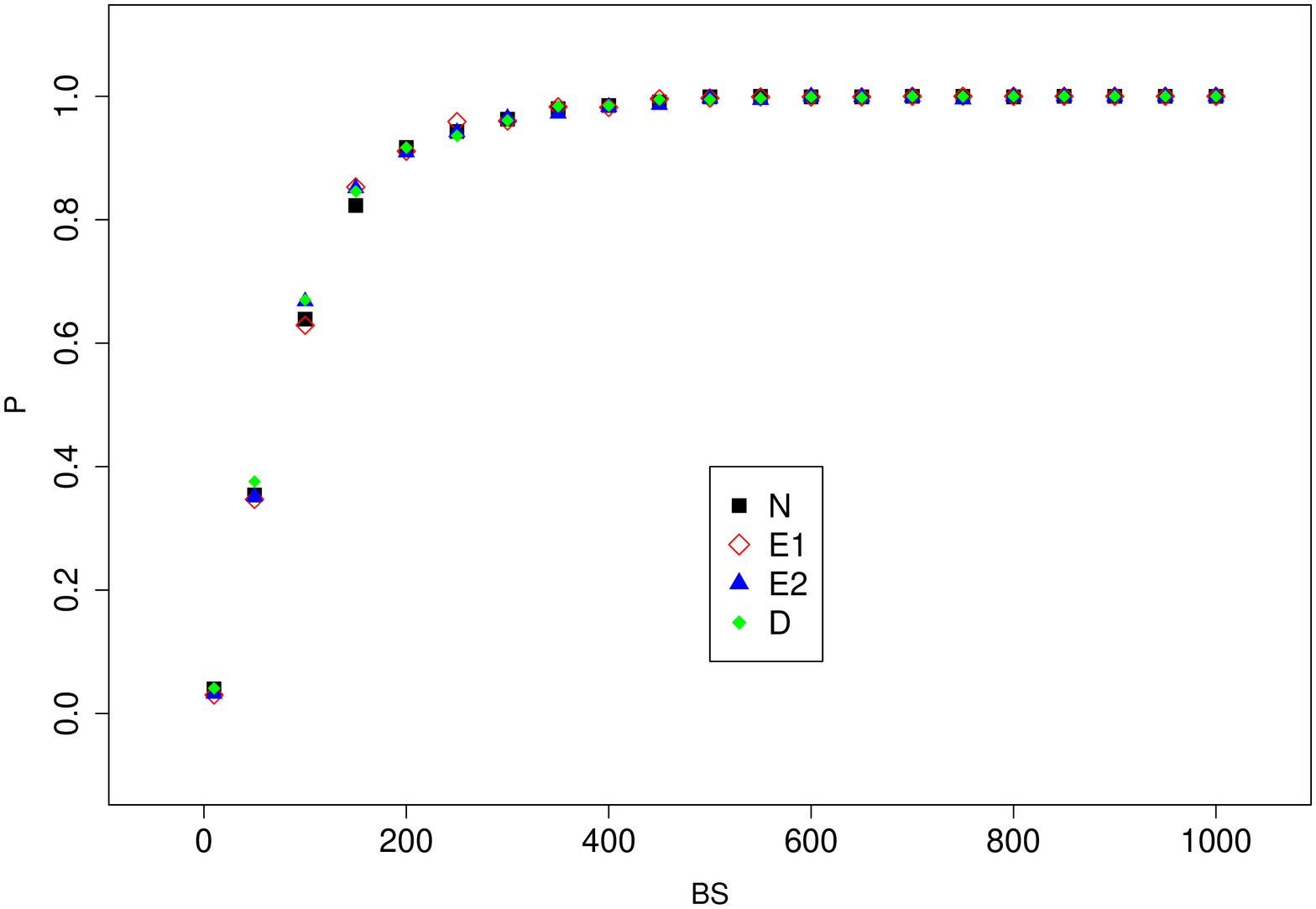}
\includegraphics[width=8cm,angle=-90]{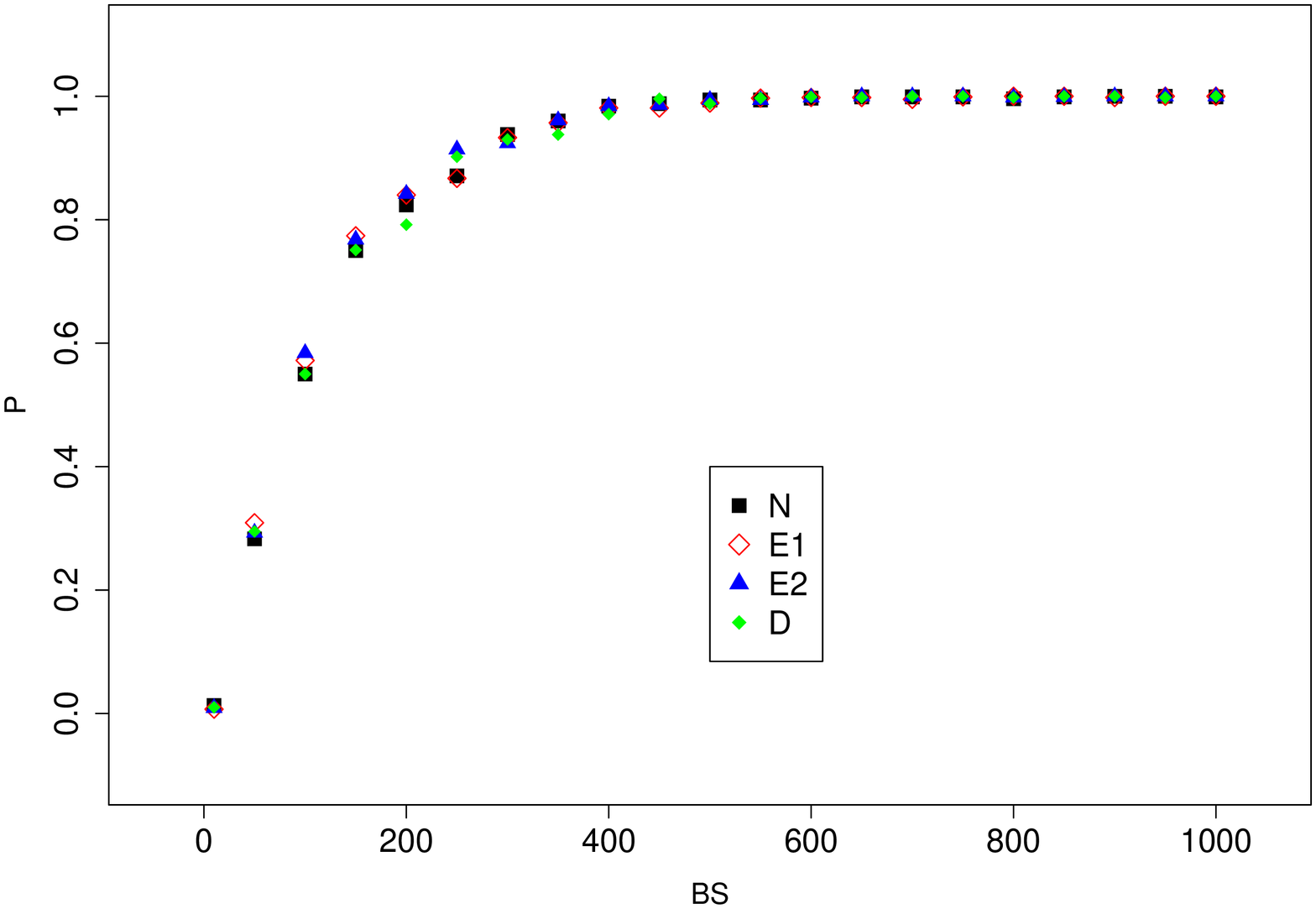}
\caption{Probability of local circulation of virus (probability of having at least
one secondary case). Top: with a tropical temperature scenario. Bottom: with a
temperated climate.\label{f1}}
\end{figure}

The study was performed in two different climatic scenarios, one with constant
temperature of 23 degrees Celsius, that represents tropical regions and one
with the mean and amplitude characteristic of Buenos Aires, a city with
temperate climate. The number of effective breeding sites \cite{oter06} was
varied between 50 and 1000.
\begin{figure}[f]
\includegraphics[width=10cm,angle=-90]{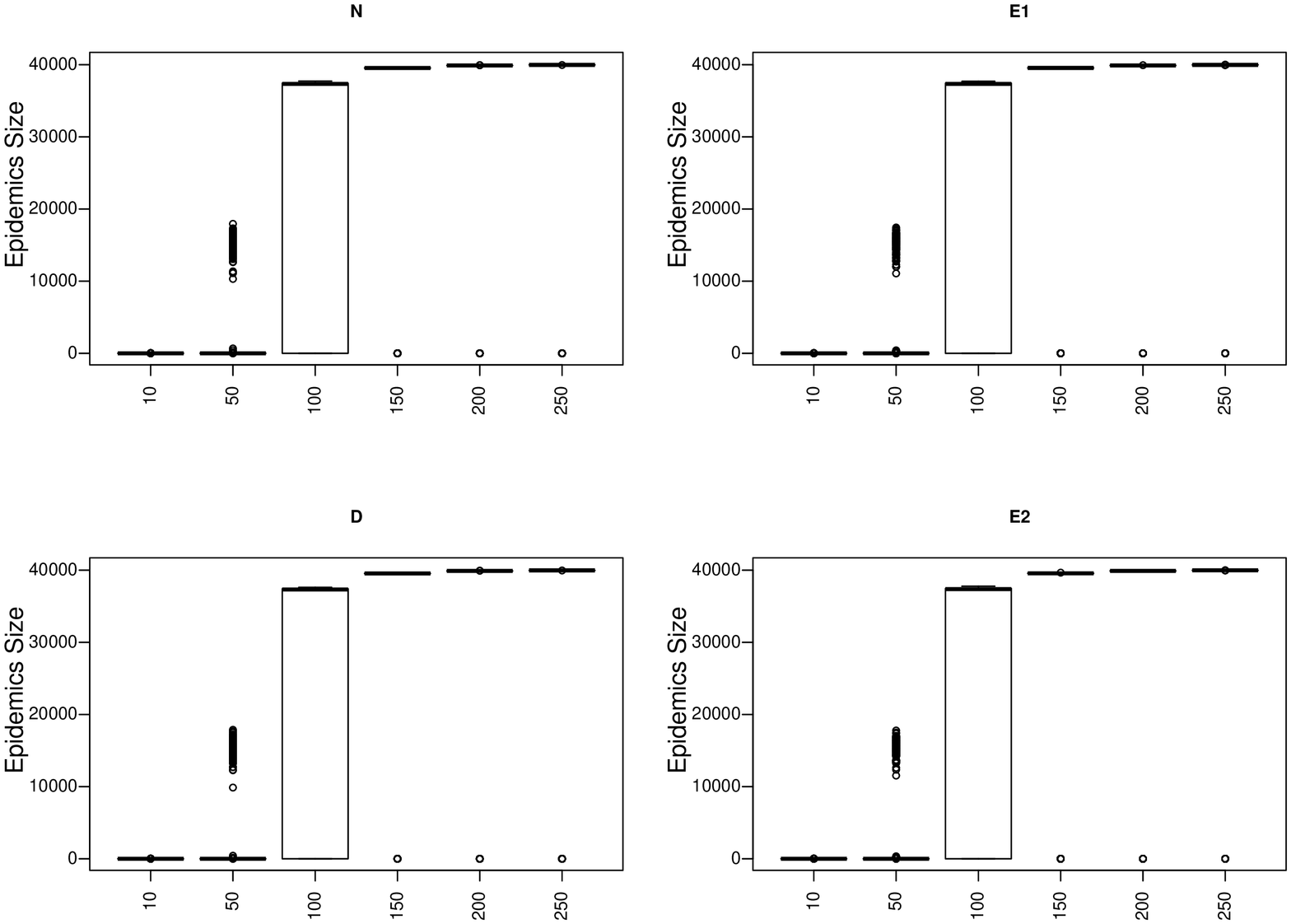}
\caption{Box plot graphs for the epidemic size (total number of infected humans)
at constant temperature. Top-left: N-distribution, top-right: E1-distribution, bottom-left
D-distribution and bottom-right: E2-distribution.\label{f2}}
\end{figure}

The main questions were: considering the total number of recovered individuals,
how is the distribution of epidemic sizes influenced by the choice of
distribution?  How does the probability for having no secondary cases change?
How is the predicted duration of the epidemic influenced by our choices? And
finally, how is the distribution of time between epidemiologically related
cases affected?

Before showing our results, it is worth to realise that the probability of
having no secondary cases will not be sensitive to the choice of distribution
for the exposed period, as this probability depends only on the probability of
the introduced case being bitten by the mosquitoes, the probability of the
mosquitoes of acquiring the virus, surviving the extrinsic period and finally
transmitting the virus to a human in a bite. Nothing in this process depends on
the choice of  distribution.  In contrast, we expect the distribution of times
between epidemiologically related cases to depend strongly on the choice of
distribution, since it reflects the sum of the two incubation periods
(intrinsic and extrinsic).

The simulations were performed using identical mosquito populations  in all
cases (same data file), for an homogeneous urban area of 20 by 20 blocks,
hosting 100 people per block, making a total of 40000 human beings. Hence, all
differences correspond to the disease dynamics that was previously identified
as the main source of stochastic variations. The simulations with temperate
climate were started on January 1st, i.e., ten days after the summer solstice
(December 21st in the southern hemisphere).  The reported statistics is
computed by averaging the outcomes of 1000 simulations with different seeds for
the pseudo-random routines.

Confirming the reasoning above, there is no sensitivity for the probability of
having local circulation of the virus (defined as having at least one secondary
case following the introduced case) with the statistical distribution of the
exposed period, see Figure \ref{f1}.

The size of the epidemic at constant temperatures makes a transition from very
small outbreaks to a large outbreak reaching almost all the people in the
simulation. The transition happens in the region 50-100 breading sites per
block for the four distributions studied, see Figure \ref{f2}. We detected no
epidemiologically important differences produced by the use of one or another
distribution function.

The size of epidemics with seasonal dependence is presented in Figure \ref{f3}.
The size of the epidemic outbreak begins to increase with the number of
breeding sites in the 100-150 region for all the distributions, suggesting that
$R_0$ (the basic reproductive number) is not affected by the statistic of
exposed times (conceptually, $R_0$ is the average number of secondary
cases produced by a single case when the epidemic starts).
The results for the D-distributions and N-distribution do not present
differences.  The E1-distribution (equal mean) overestimates the final size
while the E2-distribution substantially agrees with the N (experimental) one.
Both exponential distributions present a larger variance than the
N-distribution.  This difference may matter when worst-possible scenarios are
considered.
\begin{figure}[f]
\includegraphics[width=10cm,angle=-90]{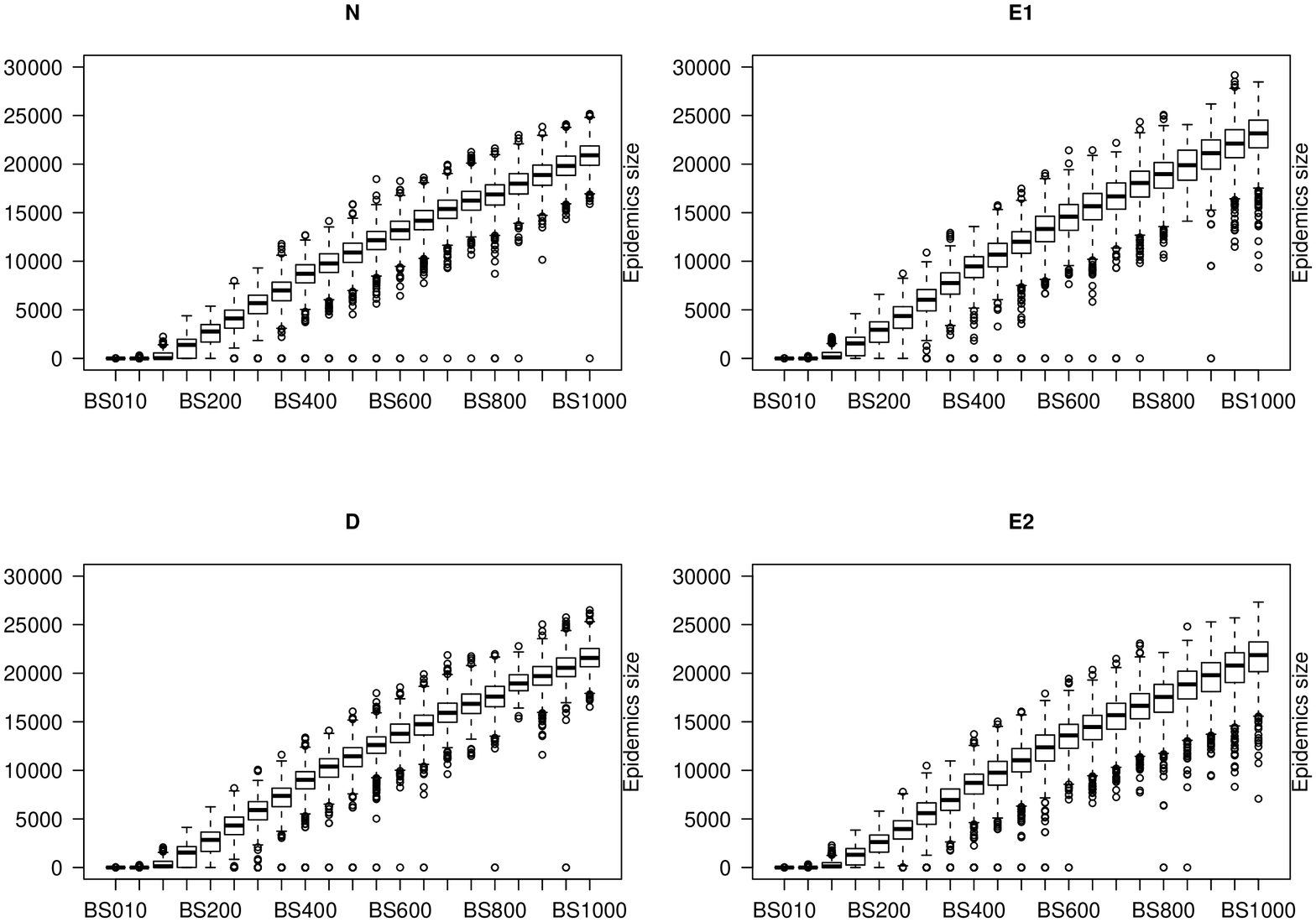}
\caption{Box plot graphs for the epidemic size with seasonal dependence.
Top-left, N-distribution, top-right, E1-distribution, bottom-left
D-distribution and bottom-right, E2-distribution.\label{f3} In x-axis 
number of breading sites, BS, per block.}
\end{figure}

It is possible to argue that the differences in epidemic size observed for
different distributions in the case with seasonal dependence correspond to a
faster evolution of the epidemic outbreak during the time-window of favourable
conditions.  On the other hand, in the constant temperature scenario the
epidemic outbreak stops because of the decrease in susceptible people produced
by the epidemic. In this case we observe no significant differences, see Figure
\ref{f4},  the epidemic size always ends up around 40000 individuals. However,
the epidemics evolves faster the larger the number of breeding sites (rightmost
plots).  It is worth observing that major outbreaks last more than one year in
this 40000 people urbanization. Since there is no human movement incorporated,
the duration of the outbreak depends critically on the dispersion of female
mosquitoes. 
\begin{figure}[f]
\includegraphics[width=2.2cm,angle=-90]{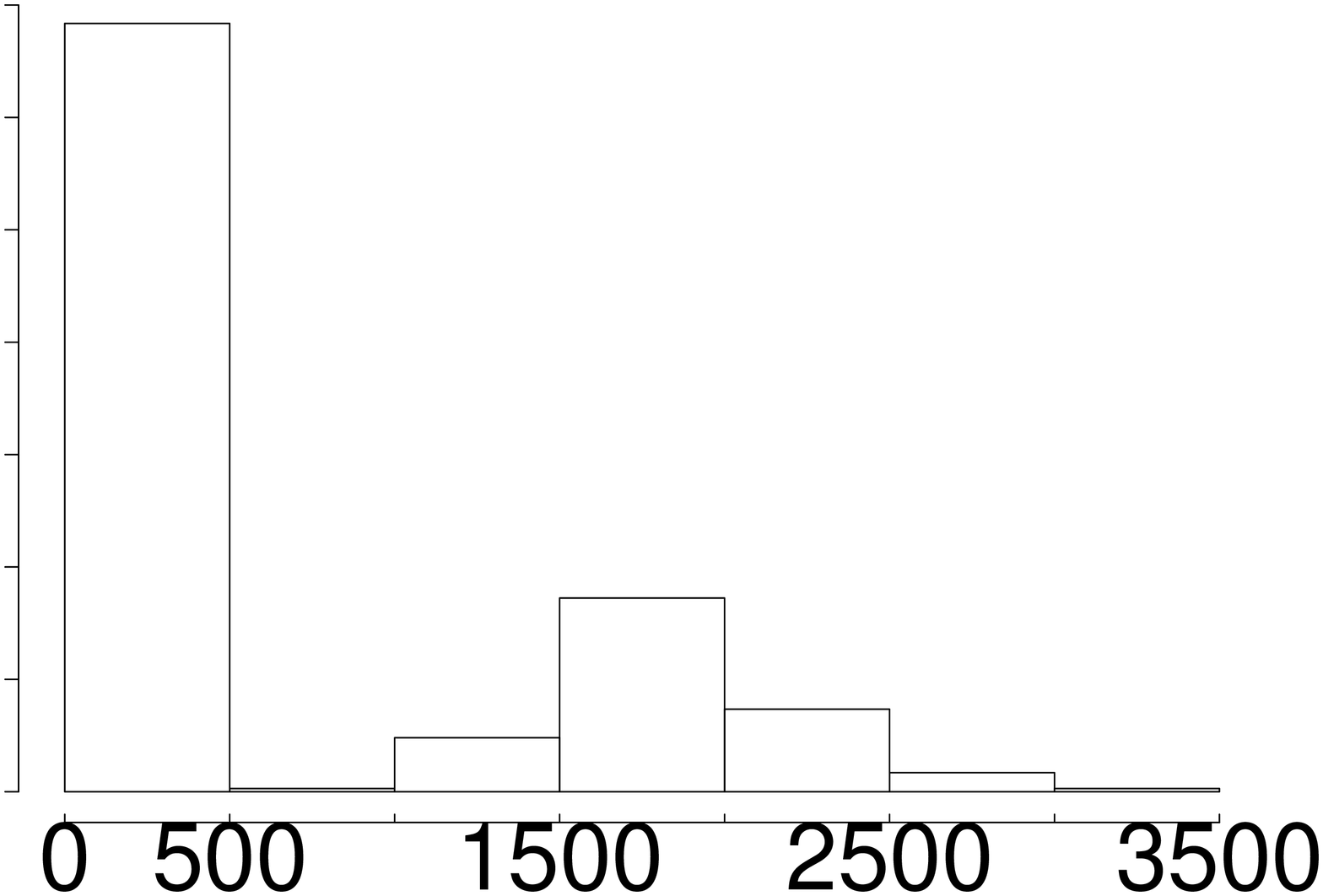}
\includegraphics[width=2.2cm,angle=-90]{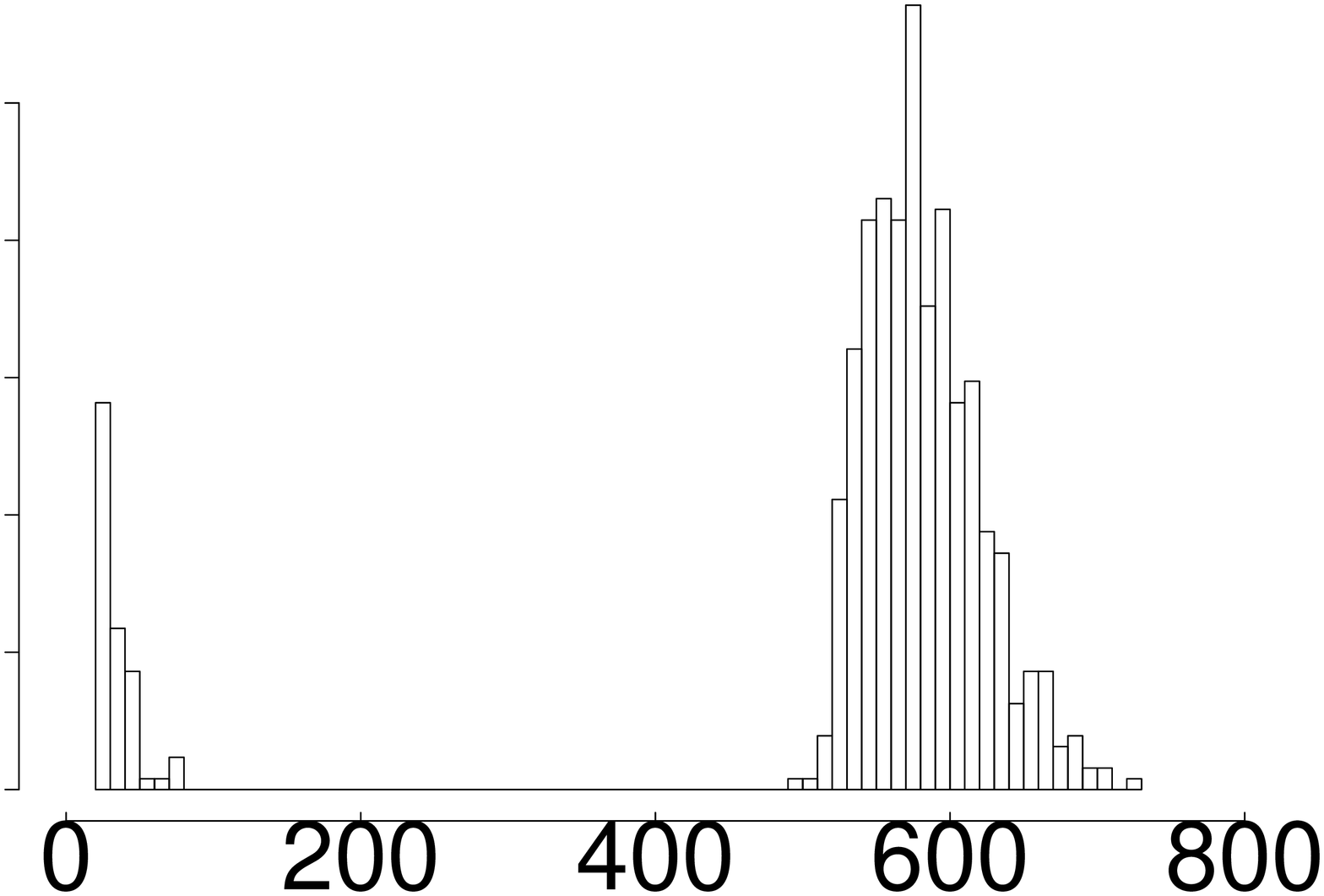}
\includegraphics[width=2.2cm,angle=-90]{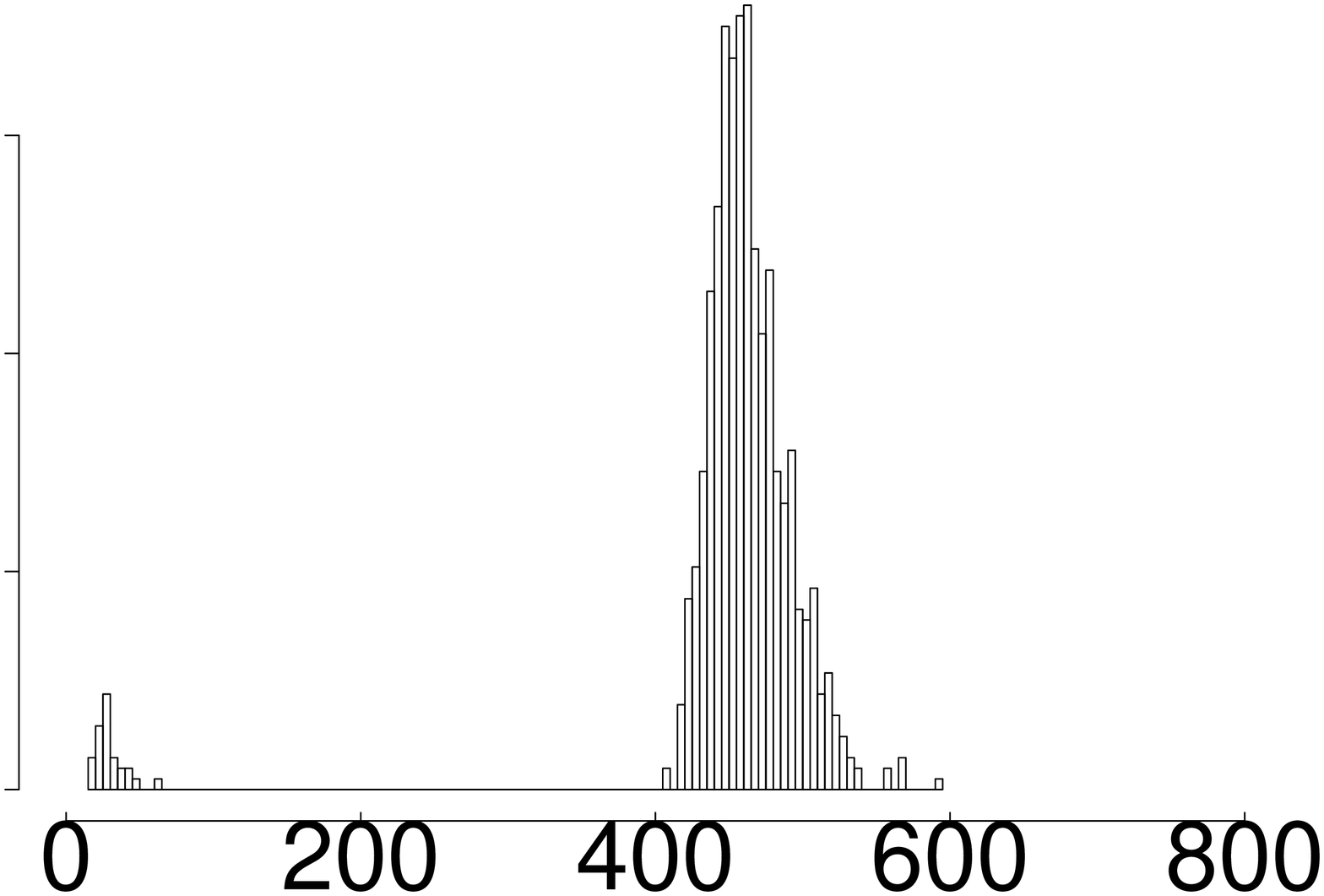}
\includegraphics[width=2.2cm,angle=-90]{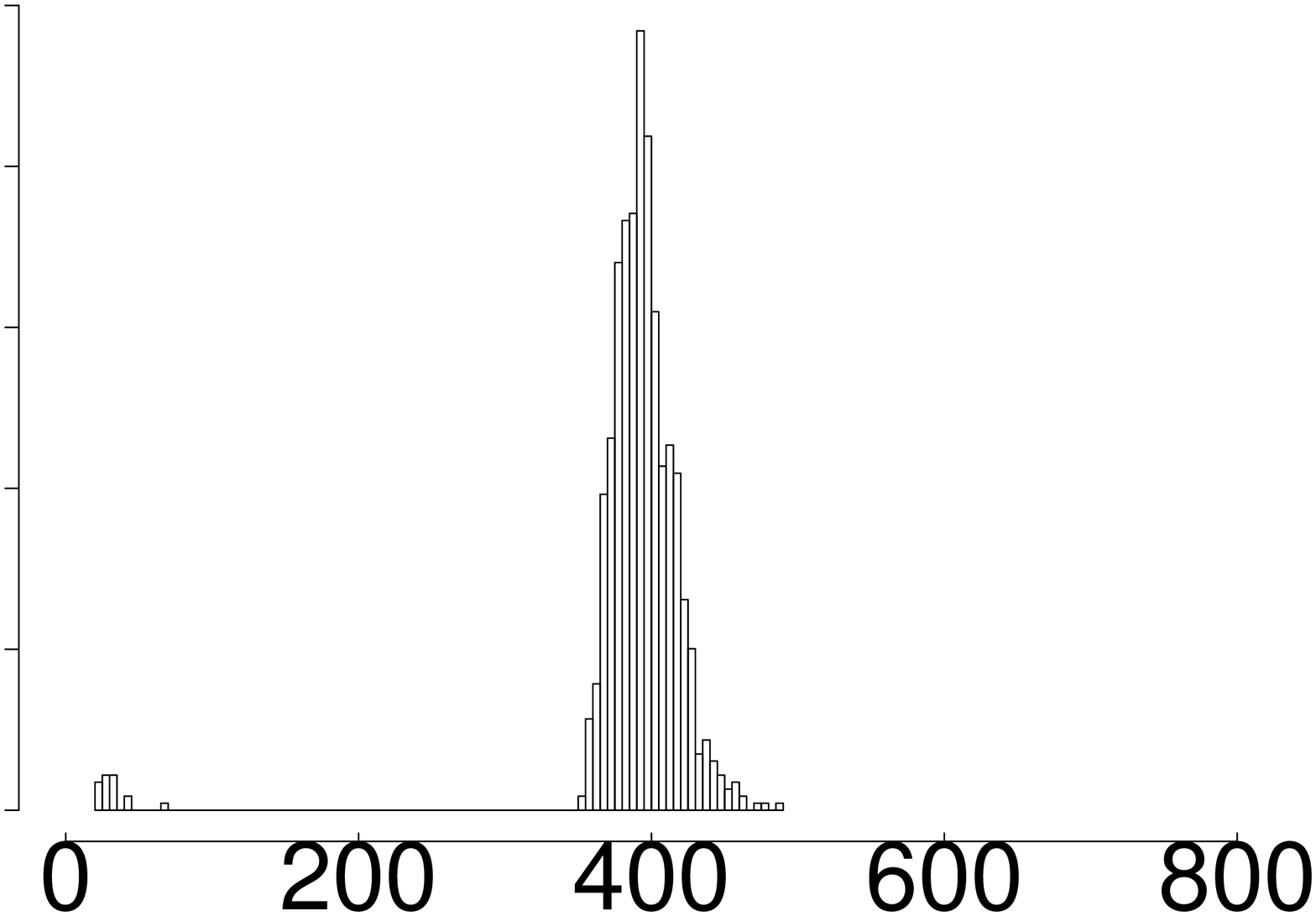} \\
\includegraphics[width=2.2cm,angle=-90]{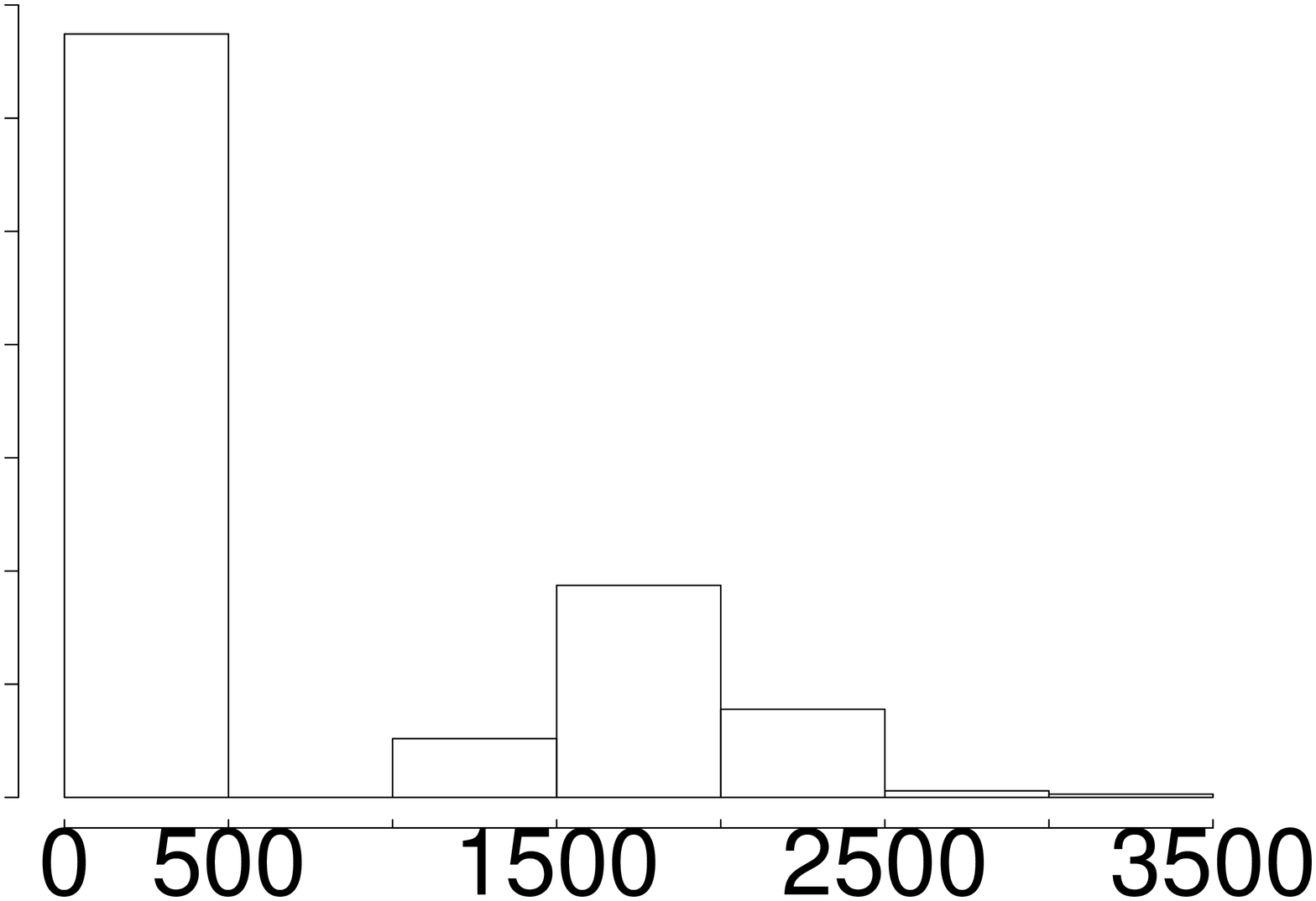}
\includegraphics[width=2.2cm,angle=-90]{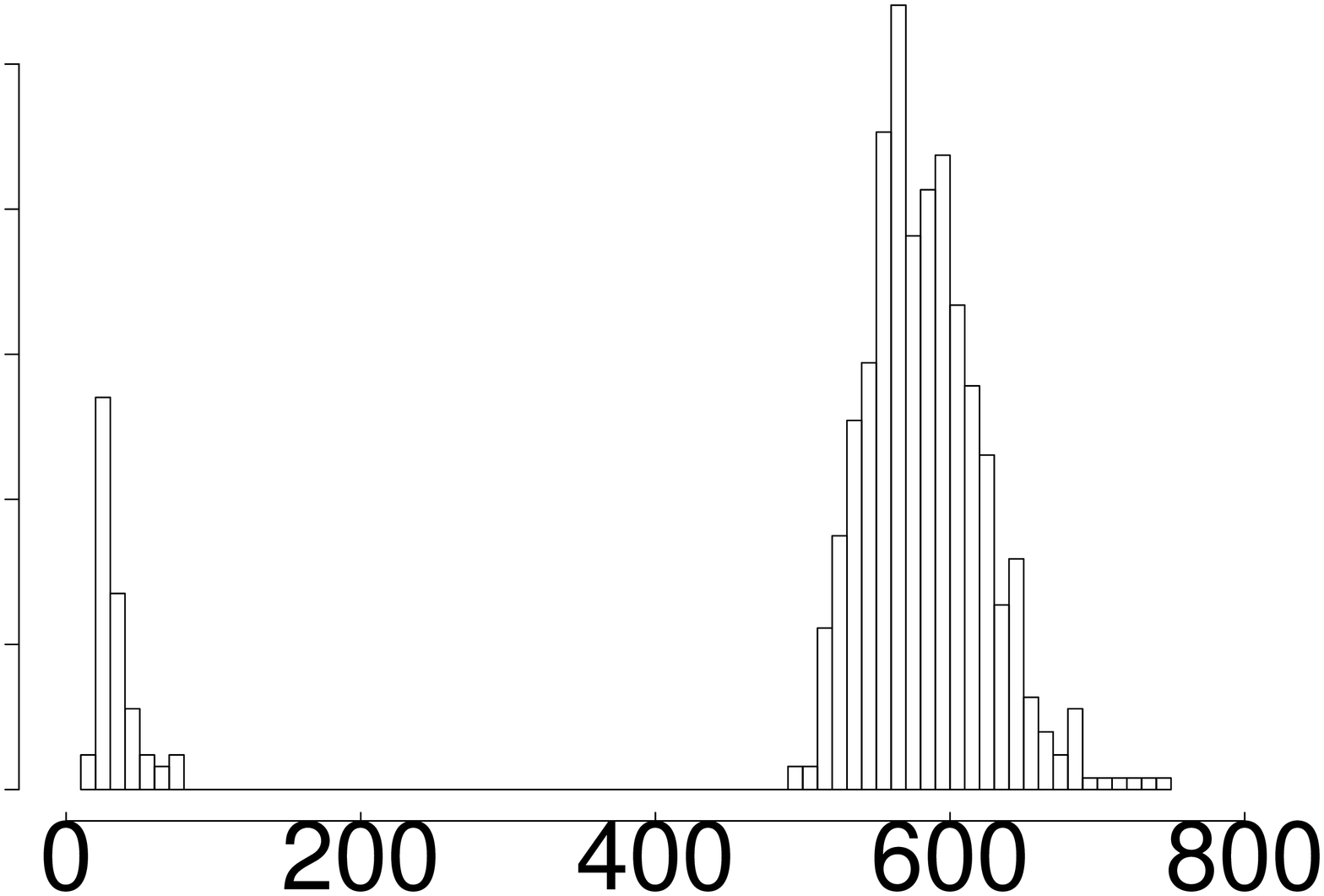}
\includegraphics[width=2.2cm,angle=-90]{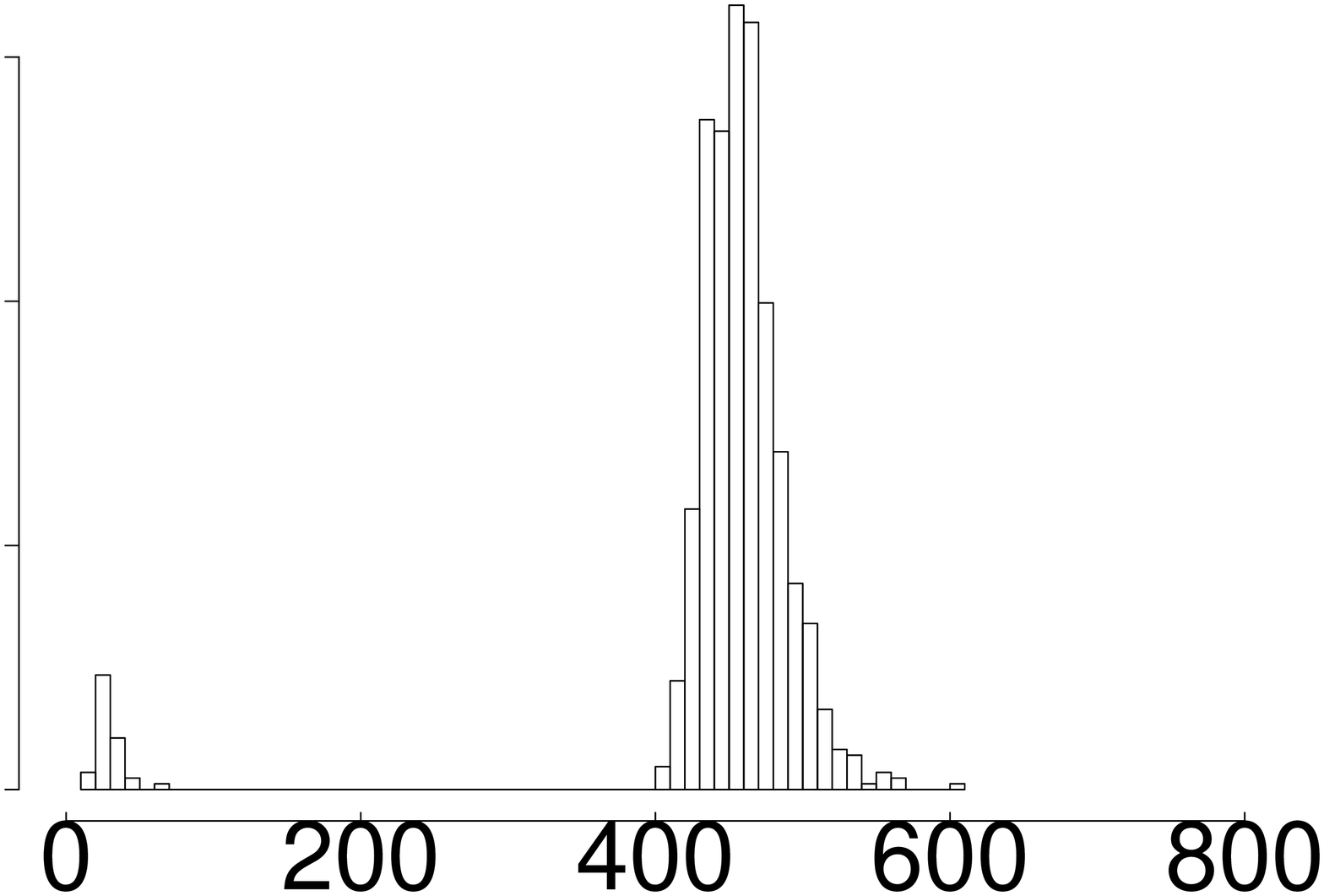}
\includegraphics[width=2.2cm,angle=-90]{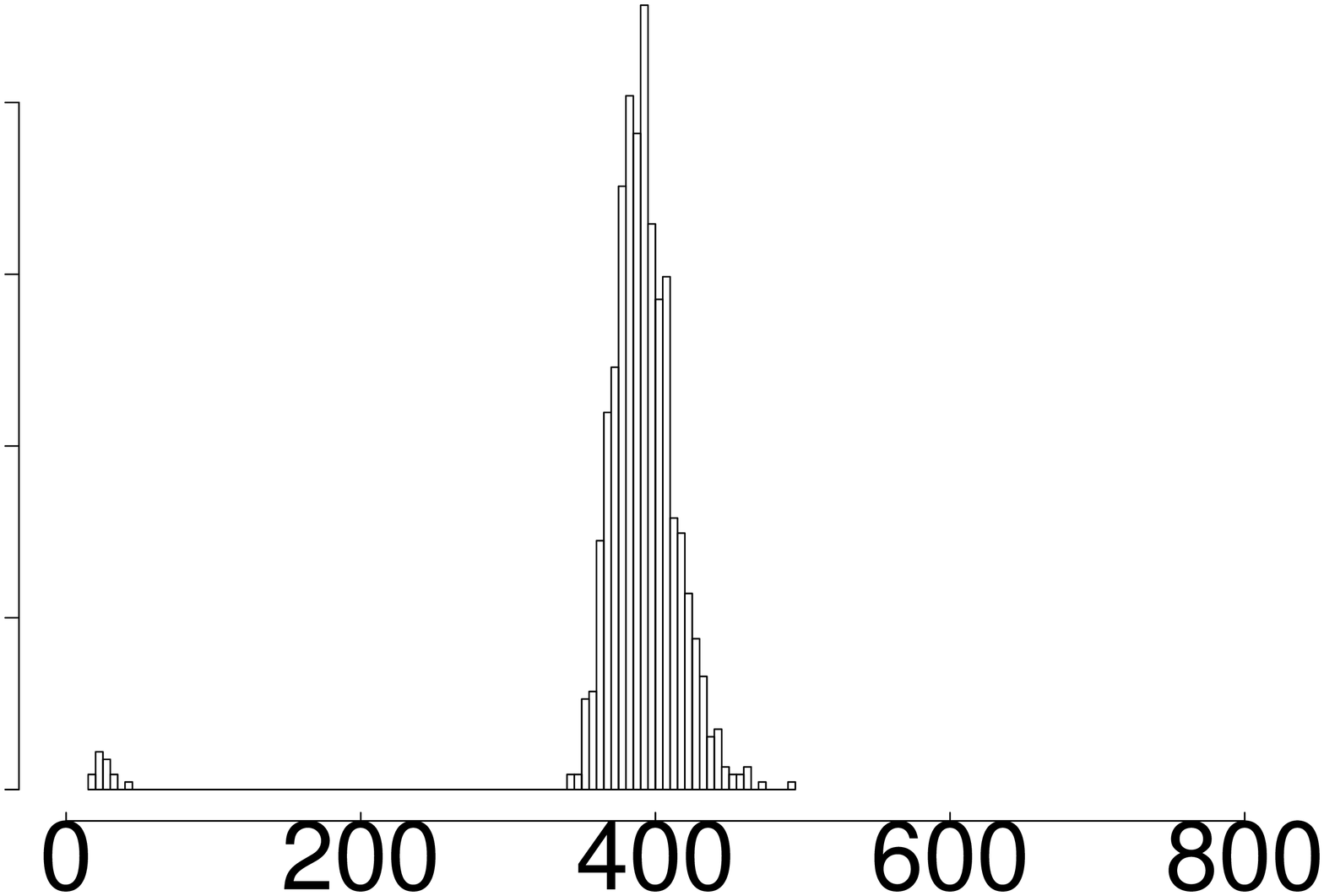} \\
\includegraphics[width=2.2cm,angle=-90]{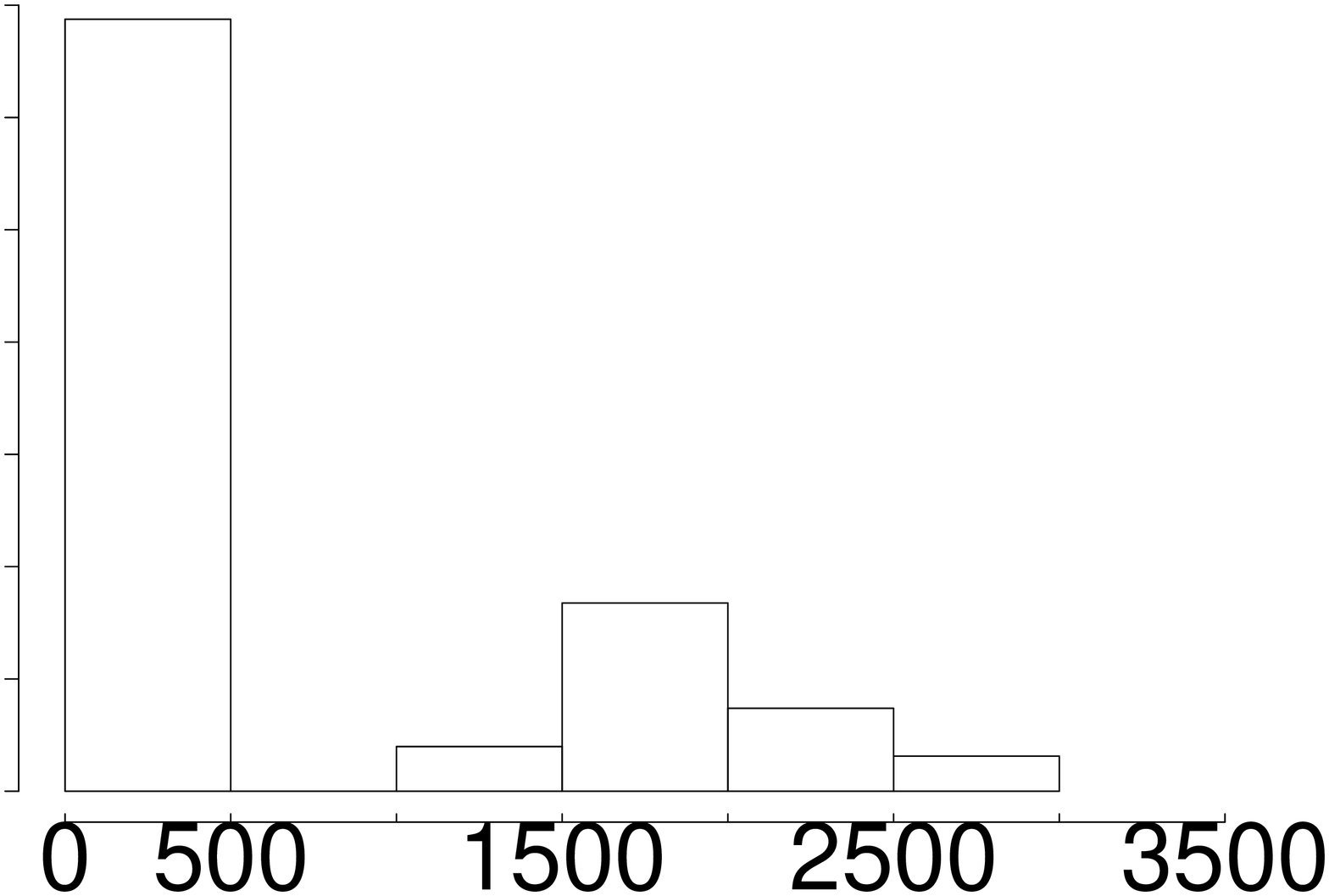}
\includegraphics[width=2.2cm,angle=-90]{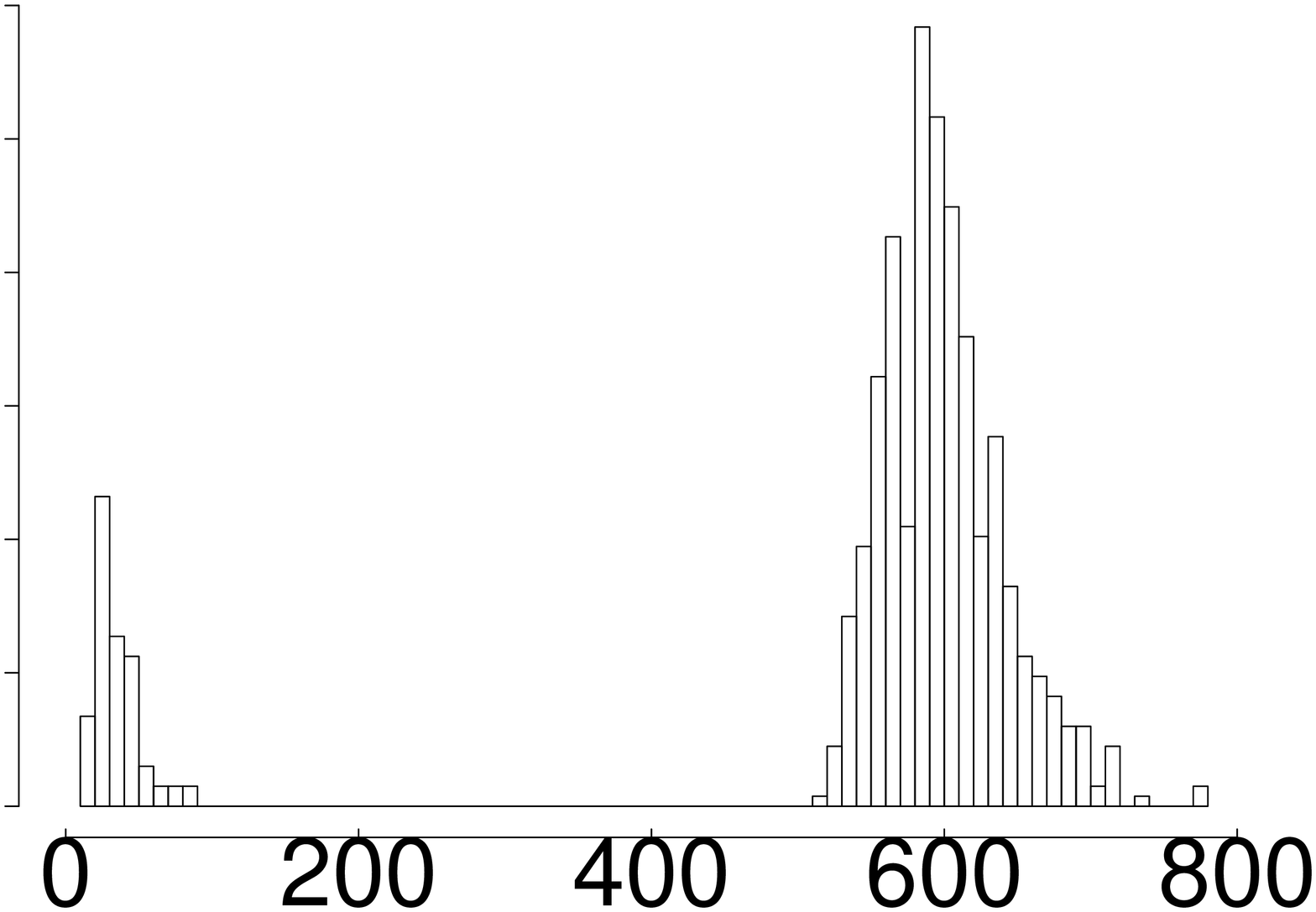}
\includegraphics[width=2.2cm,angle=-90]{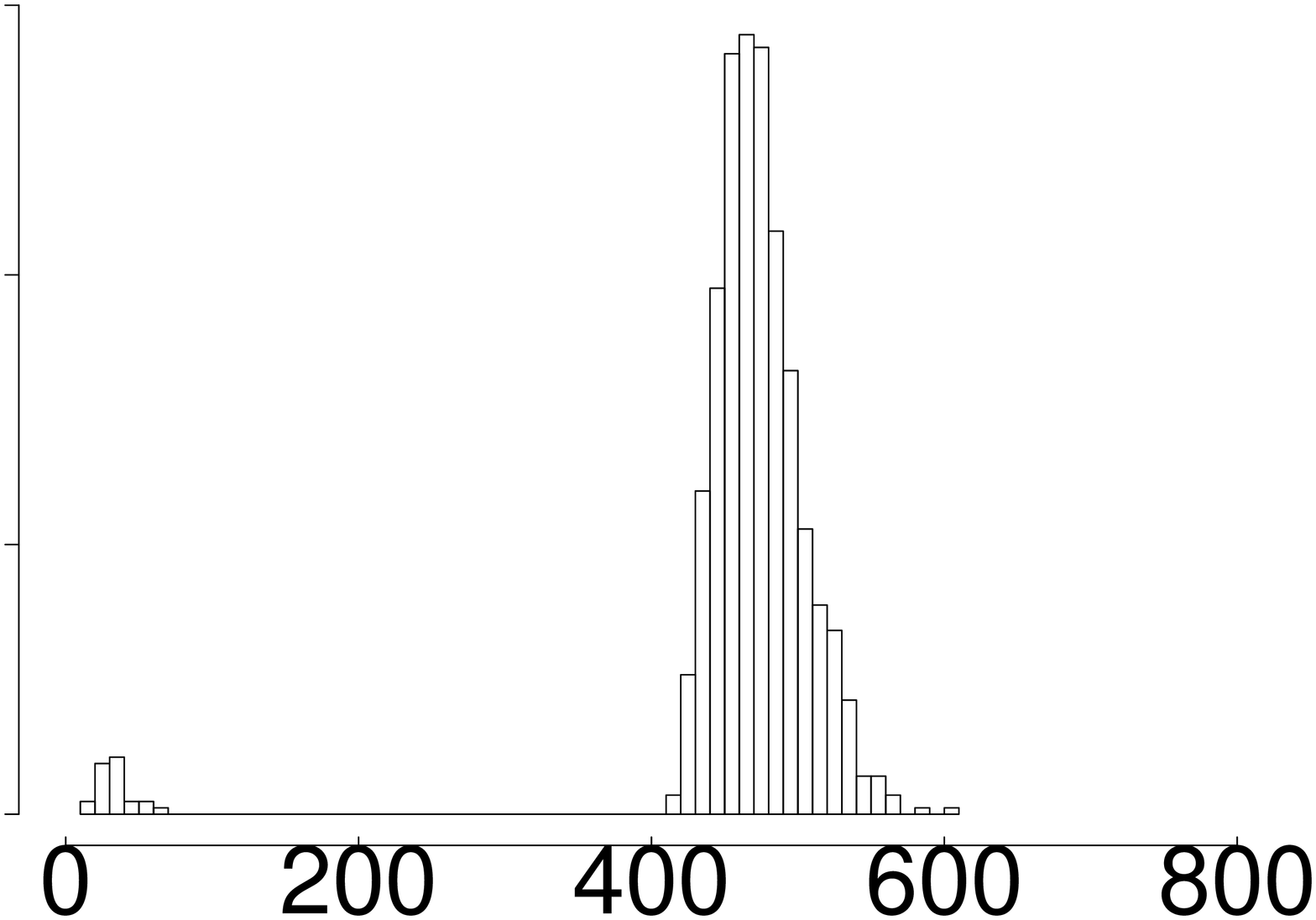}
\includegraphics[width=2.2cm,angle=-90]{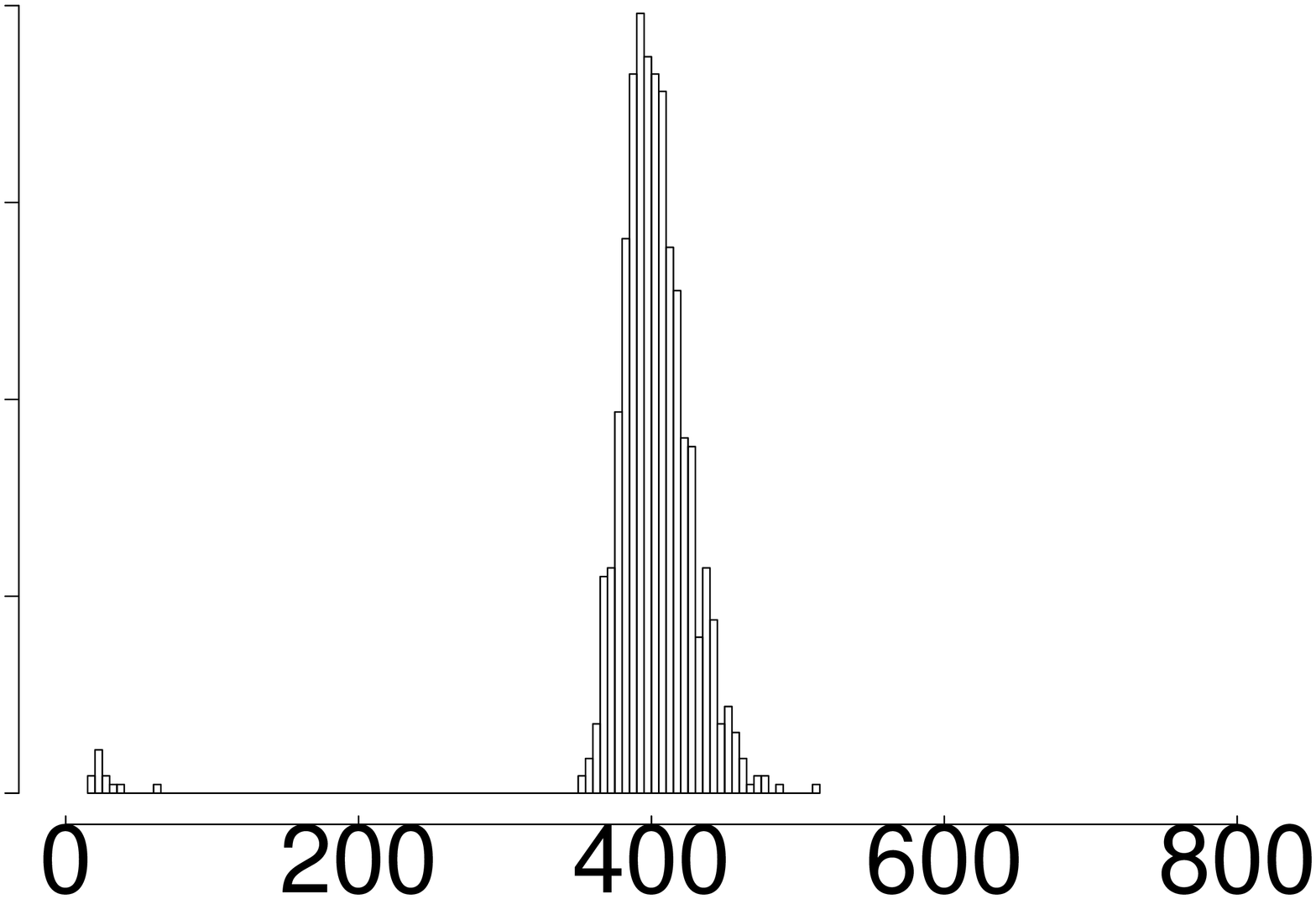} \\
\includegraphics[width=2.2cm,angle=-90]{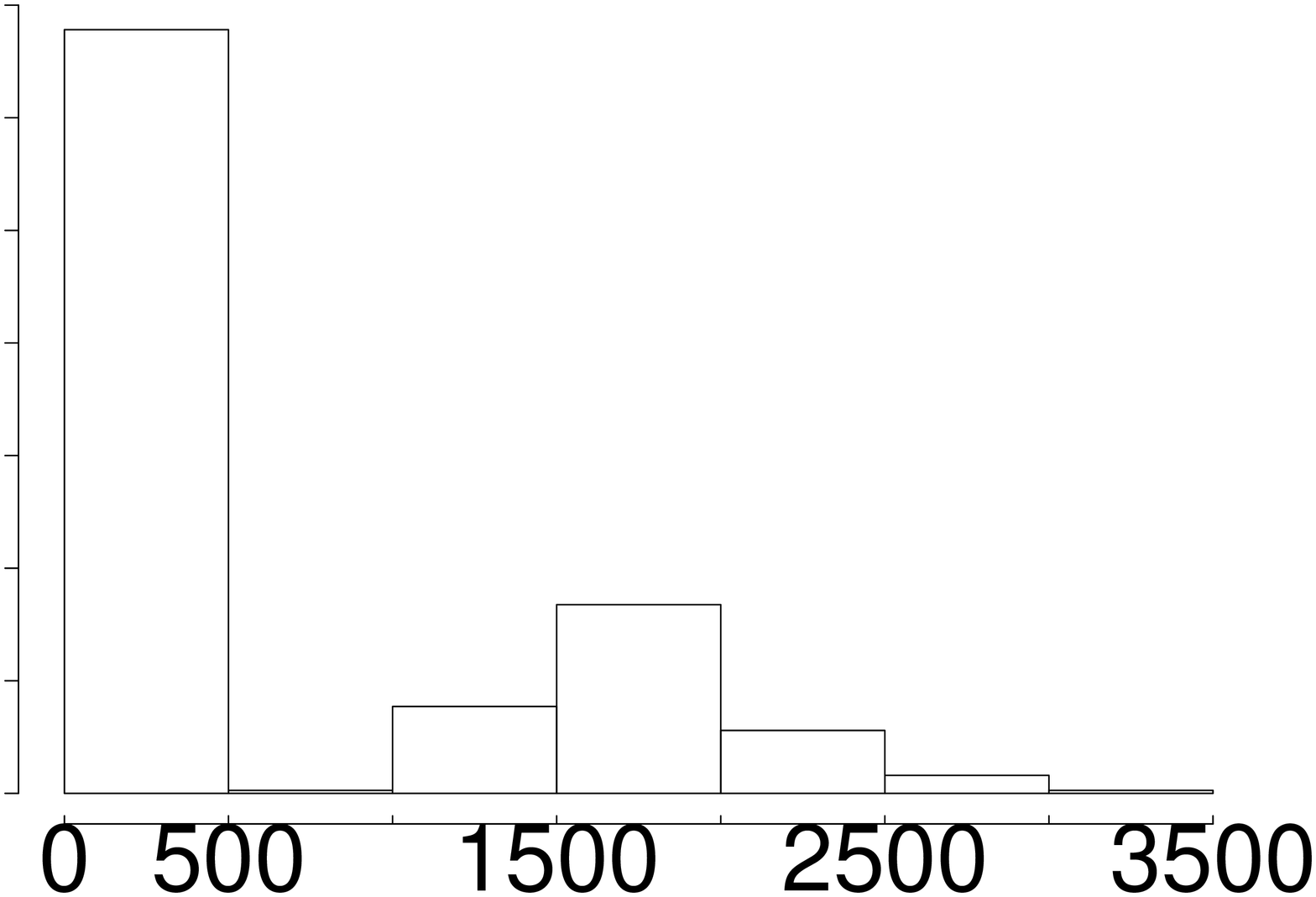}
\includegraphics[width=2.2cm,angle=-90]{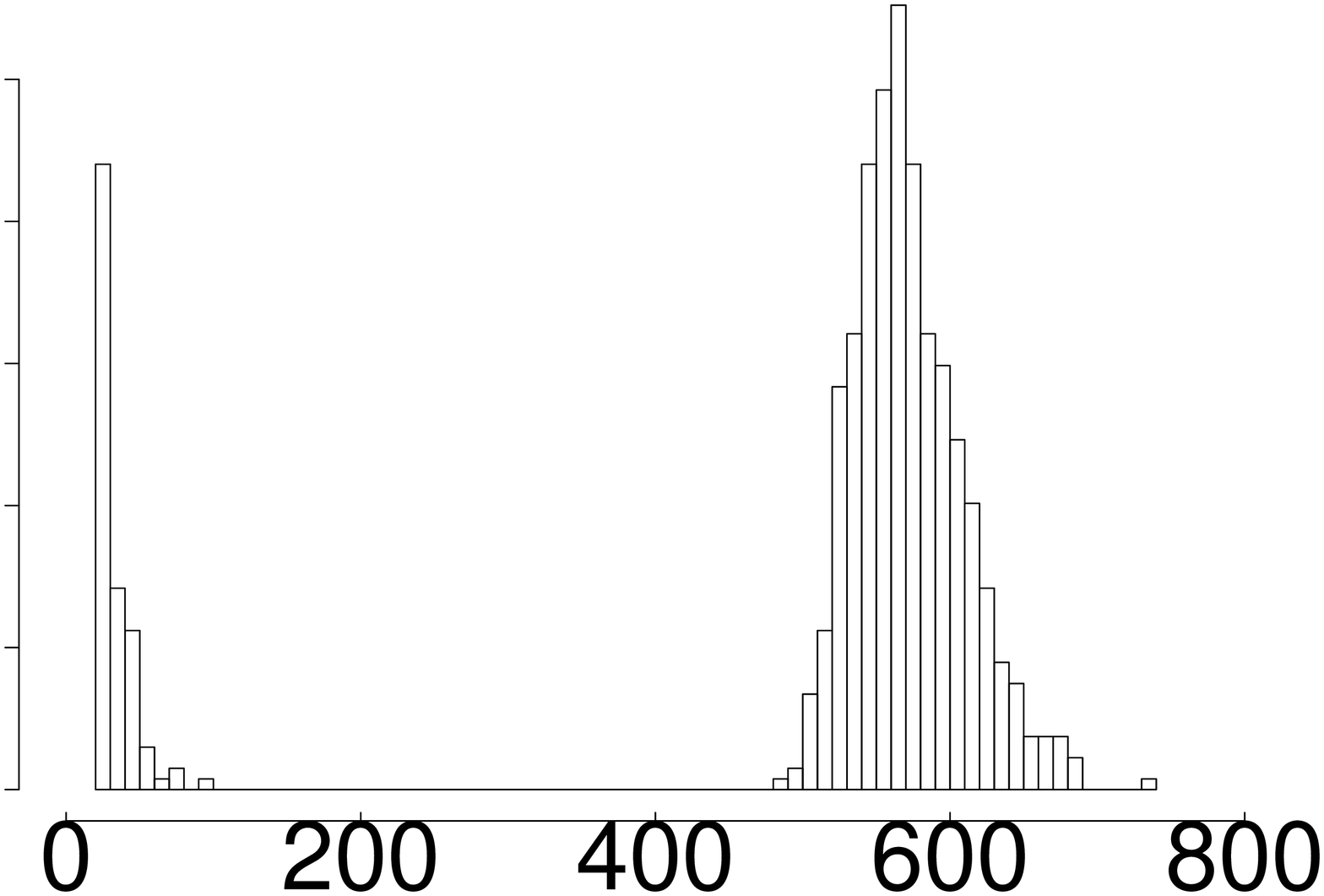}
\includegraphics[width=2.2cm,angle=-90]{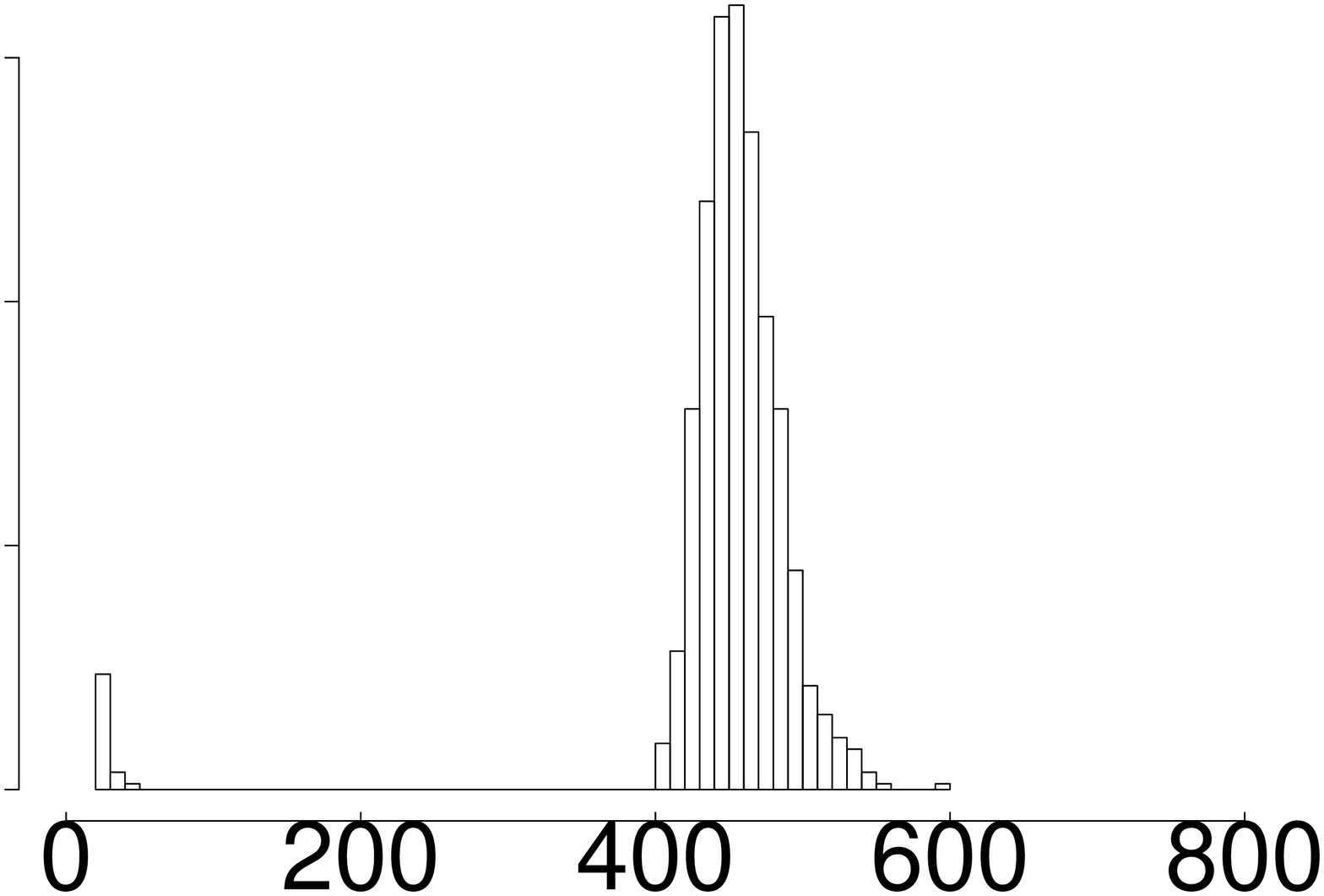}
\includegraphics[width=2.2cm,angle=-90]{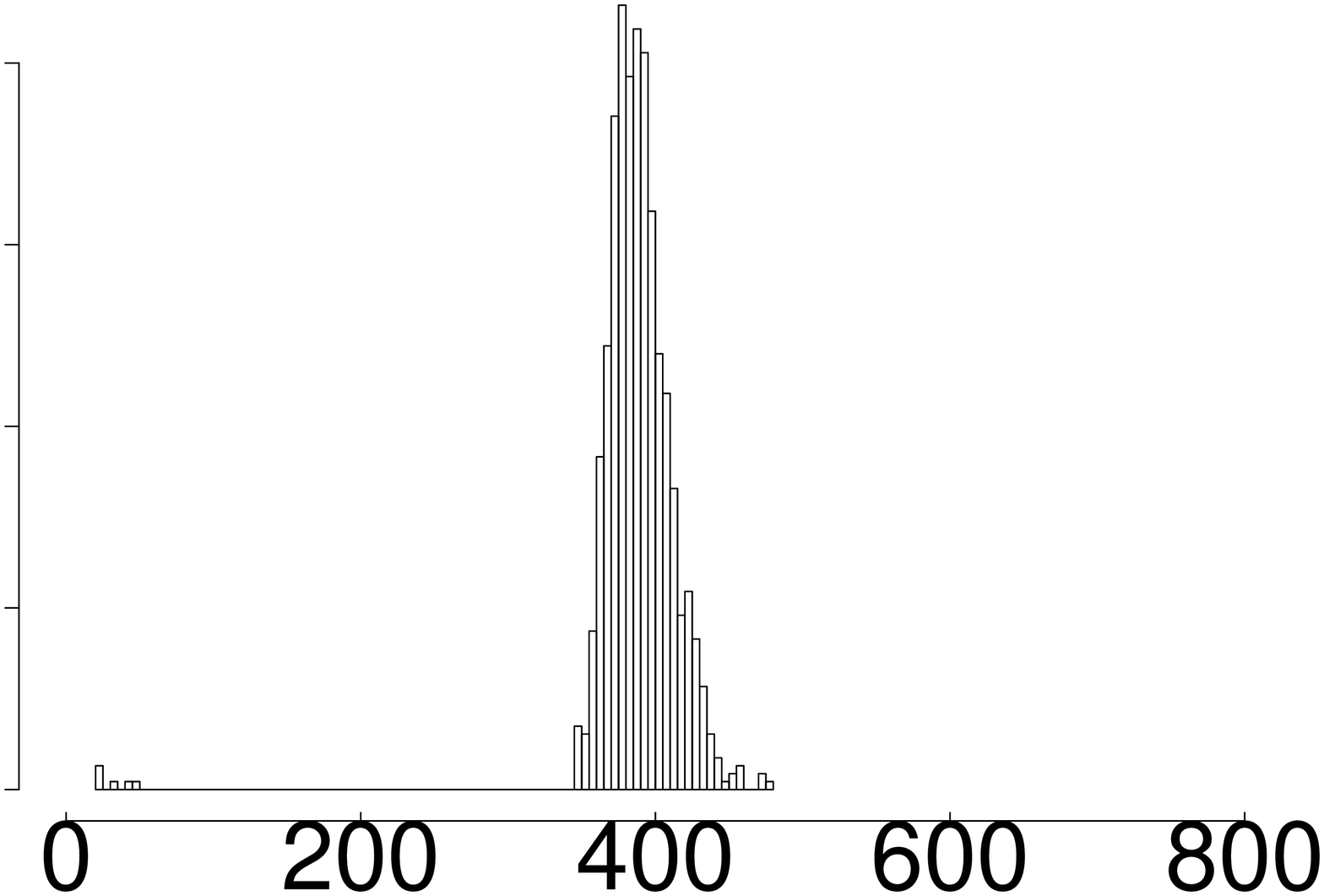} 
\caption{Statistics for the duration (in days)  of the epidemic outbreak at
constant temperature for different number of breeding sites, left to right 50,
100, 150 and 200 BS per block. Top line, N-statistics, second line
E1-statistics, third line E2-statistics and bottom line
D-statistics.\label{f4}}
\end{figure}

At this point we could ask: is there any relevant statistics that depends on
the distribution of exposed times? The answer is yes. Assume a case of imported
dengue is detected, how long do we have to wait to know if we are facing an
outbreak or not? The time elapsed between the primary and the first secondary
case is given by adding the extrinsic incubation period and the exposed time.
Box plots produced after 1000 simulations of outbreaks wit a mosquito
population supported on 200 BS per block are displayed in Figure \ref{wait}. We
observe that in this case the exponential distributions exaggerate the
dispersion of results producing too early as well as too late cases as compared
with the experimental distribution while the delta-distribution compresses the
``alert window'' too much with respect to the experimental distribution.
\begin{figure}[f]
\includegraphics[width=9cm,angle=-90]{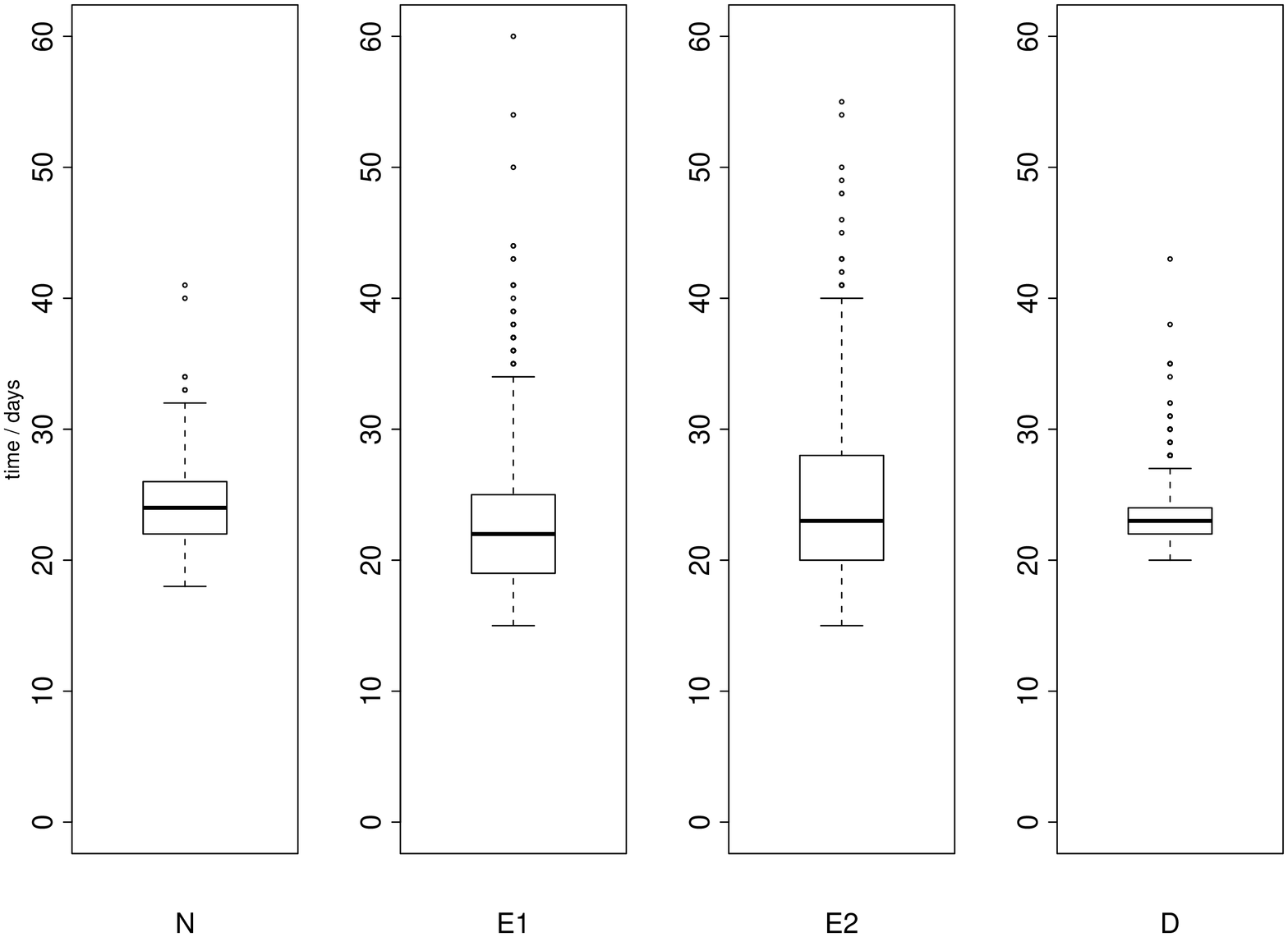}
\caption{Distribution of times between the arrival of an infectious person and
the time for the first secondary case. The scenario corresponds to 200BS/block
at constant temperature of 23ºC and 100 people per block. Results correspond to
the delta-distribution (D), the experimental distribution (N), the exponential E2
and E1 distributions.\label{wait}}
\end{figure}

\section{Is the IBM model an implementation of a compartmental model?}
The use of IBM in epidemiology is usually advocated on an
ontological perspective \cite{koop99,bian04}, we quote the argument in \cite{bian04}
\begin{quote}
    \dots the epidemiology literature has always described an infection history
as a sequence of distinct periods, each of which begins and ends with a
discrete event. The critical periods include the latent, the infectious, and
the incubation periods. The critical events include the receipt of infection,
the emission of infectious material, and the appearance of symptoms.\dots
Several implications can be drawn from these epidemiological principles and
serve as a conceptual model for an individual infection process. First, it is
most appropriate to represent the individual infection as a series of discrete
events and periods. Second, the discrete events have no duration in their own
right and only trigger the change between infection periods. Third, the
discrete periods indicate the infection status of an individual and are part of
the individual's characteristics.
\end{quote}

The evolution of dengue at the individual level is described in detail in the
literature, including experimental results \cite{nish07} which are seldom
available for other illnesses. Thus, dengue is a good case to put the thesis at
test.

The description in terms of events immediately calls for stochastic population
models, while the different human to mosquito transmission probabilities can be
handled easily introducing age structure in the infective human population,
The description of the proper probability distribution for the exposed (latent)
period is the major obstacle towards a compartmental stochastic model.

Each human individual spends $k$ days, $k\in\{1\dots \tau_E\}$, in the exposed
state before becoming infective. The probability of any individual to spend $k$
days is $P(k)$. Hence, in the group of $N$ individuals that become infected at
$d=d_0$, a portion $N(k)$ will spend $k$ days before they become infective,
where $N(k)$ is a random deviate taken from $Multinomial(N,P(1),\dots ,P(\tau_E))$, 
a multinomial distribution. This is the only information available
to the simulation. IBM generate additional structure, since each individual is
assigned to one of the classes $N(1),\cdots,N(\tau_E)$ in an entirely random
way. This apparent additional information is in fact arbitrary because of the
randomness and it is averaged out when presenting the results.  Although the
algorithm of the IBM model assigns the time spent in the exposed class to the
individual, the final cause of having a distribution of exposed times is not
known (neither to us nor to the algorithm). Different exposed times may arise
either because of differences (in virus resistance) in the exposed individuals
or because of differences in the incoming virus due e.g., to 
biological processes
concerning the development and transmission of the virus by the mosquito.

Let us define
\begin{eqnarray}
\hat{P}(k) &=& \sum_{j=k}^{\tau_E} P(j) \\
Q(k) &=& \frac{P(k)}{\hat{P}(k)} = P(j=k / j \ge k)\\
\hat{N}_k &=& N - \sum_{j=1}^k N(j) 
\end{eqnarray}

Clearly $Q(k) \le 1$ and $Q(\tau_E)=1$. Consider the numbers $E(1)=N$,
$E(k)=Binomial(E(k-1),1-Q(k))$ and $N(k)=E(k-1)-E(k)$ for $2 \le k \le \tau_E$.

We will now recall a few elementary results regarding multinomial distributions.
We write $Multi$ for the multinomial distribution
\begin{lemma}{(Multinomial splitting)}
\begin{eqnarray}
&& Multi(N, p(1),\dots, p(k-1), p(k), \dots, p(\tau_E)) =\\
&=& Multi(N, p(1),\dots, p(k-1), \hat{P}(k))\: Multi(\hat{N}(k), Q(k),\dots, Q(\tau_E) ) 
\nonumber
\end{eqnarray}
\end{lemma}
\begin{proof}
The proof is simple algebra simplifying the expression for the probabilities in
the right side of the equation.
\end{proof}
\begin{corollary}{(Binomial decomposition)}
The multinomial distribution can be fully decomposed in terms of binomial
distributions in the form
\begin{equation}
Multi(N,p(1), \dots , p(\tau_E)) = \prod_{k=1}^{\tau_E} Binomial(E(k), Q(k))
\end{equation}
\end{corollary}
\begin{proof}
Repeated application of the previous lemma is all what is needed.
\end{proof}

We now state the application of these results to our modeling problem.
\begin{theorem}{(Compartmental presentation)}
Consider $\tau_E$ compartments associated to the days $k=1,\dots ,\tau_E$
after receiving the virus from the mosquito for those humans that are not yet
infective. Let the population number in the $k$ compartment be $E(k)$, and the
number of humans that become infective on day $k$ be $N(k)$ which is a random deviate
distributed with $Binomial(E(k),Q(k))$, with the definitions given above. 
Then the exposed period that corresponds to the individuals is distributed with $P(k)$.
\end{theorem}
\begin{proof}
It follows immediately from the corollary.
\end{proof}

\begin{figure}[f]
\includegraphics[width=10cm]{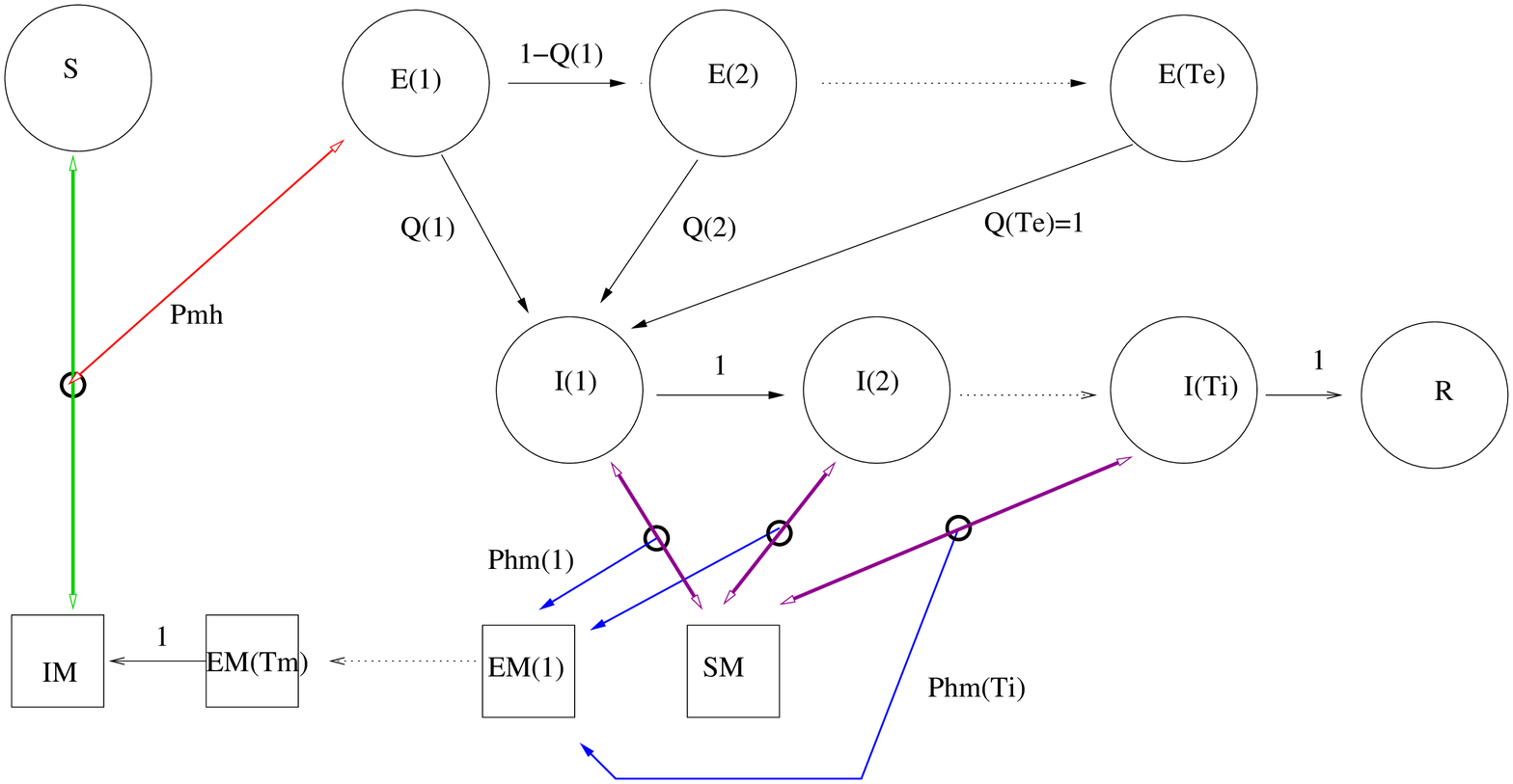}
\caption{Scheme for the progression of dengue. Circles for human subpopulations
(Susceptible, Exposed by day, Infective by day and Recovered) and squares for
mosquito subpopulations. Solid directed arrows indicate daily progression, doted
directed arrows indicate several days of progression, bi-directed arrows with a
circle indicate interactions resulting in transitions for members of one
population. Probabilities are indicated next to arrows. The mosquito dynamics
(birth, death, \dots) is not represented.\label{esquema}}
\end{figure}

\section{Summary, discussion and conclusions}
We have developed an IBM model for the evolution of dengue outbreaks that takes
information from mosquito populations simulated with an \emph{Aedes aegypti}
model and builds thereafter the epidemic part of the evolution. 
The split between mosquito evolution and epidemic evolution is not perfect 
since the events {\em bite} and {\em flight} are treated as independent events
while they are in fact correlated \cite{wolf53,reit95,muir98,edma98}, 
both being related to oviposition.

For mosquito populations sufficiently large to support the epidemic spread of
dengue, most of the stochastic variability is provided by the epidemic process
rather than by stochastic fluctuations in the mosquito population. Thus, the
splitting of the models improves substantially the performance of the codes.

The model we developed is shown to be an IBM implementation of a compartmental
model, of a form not usually considered. At the level of the description in
this work, IBM does not play a fundamental role, contrary to what it has been
previously argued  \cite{koop99,bian04}. Rather, its use is a matter
of algorithmic convenience.  This result indicates that ``IBM {\em versus}
compartmental models'' is not a fundamental dichotomy but it may be a matter of
choice (depending of the skills and goals of the user): IBM facilitate coding,
compartmental models lend themselves to richer forms of analysis.

The model was used to explore the actual influence of the distribution of
exposed time for humans in those characteristics of epidemic outbreaks that
matter the most: determining the level of mosquito abundance that makes
unlikely the occurrence of a dengue outbreak and determining the size and
time-lapse of the outbreak. The distributions used are (a) an experimentally
obtained distribution (Nishiura \cite{nish07}) (labeled N), (b) and (c)
exponential distributions adjusted to have the same mean or the same median as
N, labeled E1 and E2, and (d) a fixed time equal to the experimental mean
(labeled the D-distribution).  The probability of producing one or more
secondary cases after the arrival of an infective human does not depend on the
choice of distribution. The characteristic size of the epidemics under a
temperate climate are exaggerated by the E1-distribution but presents no
substantial difference for the other distributions. The dispersion of values is
exaggerated by both exponential distributions. We observed no important
differences in the duration of epidemics developed under a constant temperature
since the outbreaks reach almost all the population and the velocity is
regulated by the dispersion of the mosquitoes in the absence of movement by
humans.

The only statistic able to discriminate easily between the four distributions
of exposed time was found to be the time of appearance of the first secondary
case. A result that was expected as well.

In conclusion, only very specific matters seem to depend on the
characteristics of the distribution of 
exposed times for human beings.
Looking towards the past, conclusions reached using exponential and delta
distributions cannot be objected on such basis. Looking towards the future,
simple compartmental models can be constructed as well using realistic
distributions and there is no reason to limit the models to the choice:
exponential, gamma or delta.

\bibliographystyle{elsarticle-num}

\newpage
\listoffigures

\end{document}